\documentclass[12pt]{article}

\usepackage[T1]{fontenc}
\usepackage[active]{srcltx}
\usepackage{lmodern}
\usepackage{amsmath,amstext,amssymb,graphicx,graphics,color}
\usepackage{mathtools}
\usepackage{todonotes}
\usepackage{tikz}
\usepackage{mleftright}
\usepackage{pgfplots}
\usepackage{theorem}
\usepackage{cite}               
\usepackage{hyperref}           
\usepackage{bbm}                
\usepackage{fancybox}           
\usepackage[section]{placeins}  
\usepackage[font=small,margin=1em]{caption}
\usepackage{subcaption}
\usepackage{enumitem}
\usepackage[margin=2.75cm]{geometry}
\theorembodyfont{\normalfont\slshape} 
\newtheorem{definition}{Definition}
\newtheorem{theorem}{Theorem}
\newtheorem{proposition}{Proposition}[section]

\newtheorem{lemma}[proposition]{Lemma}
\theoremstyle{break} 

\newenvironment{proof}%
{{\par\noindent \bf Proof. \nobreak}}%
{\nobreak \removelastskip \nobreak \hfill $\Box$ \medbreak}
{{\par\noindent \bf Proof \nobreak}}%
{\nobreak \removelastskip \nobreak \hfill $\Box$ \medbreak}
{{\par\noindent \bf Proof lemma. \nobreak}}%
{\nobreak \removelastskip \nobreak \bf End proof lemma. \medbreak}
\newenvironment{remark}{\par \medskip \noindent {\bf Remark. }\nobreak}{\par \medskip}

\usepackage{fancyhdr}
 \fancypagestyle{plain}{
  \fancyhead{}
  
}

\def\paragraph#1{{\bf #1\ }}
\newcommand{\RN}[1]{%
  \textup{\uppercase\expandafter{\romannumeral#1}}%
}
\newcommand{\expo}{\mathrm{e}}

\newcommand{\dd}{\mathrm{d}}
\newcommand{\EE}{\mathbb{E}}

\newcommand{\cP}{\mathcal{P}}
\newcommand{\cA}{\mathcal{A}}
\newcommand{\cS}{\mathcal{S}}
\newcommand{\overbar}[1]{\mkern 1.5mu\overline{\mkern-1.5mu#1\mkern-1.5mu}\mkern 1.5mu}
\usepackage[utf8]{inputenc}
\usepackage{unicode_character}

\title{The incremental voter model: mean-field analysis and convergence to equilibrium}

\author{Fei Cao \footnotemark[1] \and Xiaoqian Gong \footnotemark[1]}

\begin{document}
\maketitle

\footnotetext[1]{Amherst College - Department of Mathematics, Amherst, MA 01002, USA}

\tableofcontents

\begin{abstract}
We introduce the incremental voter model (IVM), a discrete-opinion multi-agent system where agents undergo step-wise transitions biased by the opinion of a randomly selected persuader. Our incremental voter model comprises a large population of interacting agents, each holding an opinion represented by an element of the discrete set $\{-k,\ldots,0,\ldots,k\}, k \in \mathbb{N}_{+}$. At each update step as time progresses, a pair of distinct agents are selected independently and uniformly at random from the population, and the first agent (viewed as the ``listener'') updates its opinion based on that of the second (viewed as the ``persuader''), adopting a new opinion that differs from its current one by at most one unit.
By deriving the mean-field system of nonlinear ordinary differential equations (ODEs) that governs the large-population limit of the agent-based model, we develop a rigorous mathematical framework to study the asymptotic behavior of the opinion distribution in the mean-field limit. These results contribute to a deeper understanding of social influence processes in complex systems, particularly in modeling opinion polarization, and may guide the formulation of more advanced models in future research.
\end{abstract}


\section{Introduction}\label{sec:sec1}
\setcounter{equation}{0}

In recent years, the study of opinion dynamics has garnered growing interest due to its broad applicability in domains such as political science, digital culture, and public health. The mathematical investigation of how individuals form and adjust opinions within a population can be traced back to at least the 1960s. Since then, the infusion of methods from statistical physics into the social sciences has led to powerful new frameworks for modeling complex collective behavior, giving rise to the interdisciplinary fields of sociophysics and econophysics \cite{bennaim_2005,toscani_2006}. Pioneered by Galam, Gefen, and Shapir in \cite{galam_gefen_shapir_1982}, sociophysics seeks to understand the dynamics of social interaction through probabilistic models and dynamical systems.

Since the early 2000s, this field has expanded rapidly, fueled by the development and rigorous analysis of several influential models \cite{deffuant_mixing_2000,hegselmann_opinion_2002,sznajd_opinion_2000}. Among the most prominent are the Deffuant (or bounded confidence) model \cite{deffuant_mixing_2000}, the Krause–Hegselmann model \cite{hegselmann_opinion_2002}, and the Sznajd model \cite{sznajd_opinion_2000}, each of which has inspired numerous extensions. The literature on opinion dynamics continues to grow \cite{cao_iterative_2024,castellano_statistical_2009,jabin_clustering_2014,naldi_mathematical_2010,sen_sociophysics_2014}, reflecting sustained interest in understanding mechanisms of consensus formation, polarization, and clustering.

The present study is principally inspired by the recent work \cite{cao_fractal_2025} and is further informed by the iterative persuasion-polarization framework introduced in \cite{cao_iterative_2024}, where a novel stochastic agent-based opinion model is analyzed using a combination of probabilistic and analytical tools.


Our proposed agent-based opinion dynamics is an interacting multi-agent system consisting of a group of $N$ agents (with $N \in \mathbb{N}_+$), in which agents engage in pair-wise interaction and adjust their opinions accordingly. Specifically, in our incremental voter model, we consider a population of $N$ indistinguishable agents, each holding an opinion represented by an integer valued in a prescribed admissible opinion space $\Omega \coloneqq \{-k,\ldots,0,\ldots,k\}$, where $k \in \mathbb{N}_+$ is some fixed model parameter. We denote by $X^{i,N}_t$ the opinion of agent $i$ at time $t\geq 0$, and sometimes we suppress the subscript $t$ for notational simplicity. From a practical perspective, we may interpret $X^{i,N}$ as the political orientation/stance of agent $i$, with $-k$ and $k$ representing extreme left-wing and right-wing opinions, respectively. The dynamics of our agent-based opinion model are detailed as follows:

\begin{itemize}
\item At each random time generated by a Poisson clock with rate $N/2$, a pair of distinct agents $(i,j) \in \{1,\cdots,N\}^2$ with $i\neq j$ is selected uniformly at random and independently of the past. In each interaction, agent $j$ serves as the ``persuader'' and agent $i$ as the ``listener'', indicating that agent $j$ attempts to influence the opinion state of agent $i$.
\item With probability $\frac{k + X^{j,N}}{2k}$, agent $i$ increases its opinion by one unit, provided that $X^{i,N} < k$; with the complementary probability $\frac{k - X^{j,N}}{2k}$, agent $i$ decreases its opinion by one unit, provided that $X^{i,N} > -k$.
\end{itemize}
Mathematically, if the pair of agents $(i,j)$ are selected to interact at time $t$, then the opinion of agent $i$ will be updated according to
\begin{equation}
\label{eq:dynamics}
X_t^{i,N} = \begin{cases}
\left(X_{t^-}^{i,N} + 1\right)\mathbbm{1}\{X_{t^-}^{i,N} < k\} &~~ \textrm{with probability}~~ \frac 12 + \frac{X_{t^-}^{j,N}}{2k}, \\
\left(X_{t^-}^{i,N} - 1\right)\mathbbm{1}\{X_{t^-}^{i,N} > -k\} &~~ \textrm{with probability}~~ \frac 12 - \frac{X_{t^-}^{j,N}}{2k},
\end{cases}
\end{equation}
where \(\mathbbm{1}\) is the indicator function.
Loosely speaking, the update rule \eqref{eq:dynamics} biases agent $i$'s opinion based on agent $j$'s opinion level. This results in an asymmetric, directional influence model, where agents tend to move in the direction of other agents' current opinion. In other words, more extreme agents are more persuasive/influential. To highlight the novelty and the key features of our incremental voter model, we emphasize that while several classical models in the literature explore opinion dynamics with discrete opinion spaces and influence mechanisms, none of them align precisely with the structure and dynamics of our model:

\begin{itemize}
\item In a typical \textbf{voter model} \cite{clifford_model_1973,holley_ergodic_1975,heinrich_conformity_2025}, agents adopt the opinion of a randomly selected neighbor, leading to consensus over time. This model lacks the probabilistic bias based on the influencer's opinion magnitude present in our model.
\item The \textbf{Sznajd model} \cite{sznajd_opinion_2000,sznajd_review_2021} focuses on group influence where a pair of agreeing agents can convince their neighbors. The influence is based on agreement rather than the magnitude of opinions.
\item  In several \textbf{bounded confidence models} \cite{deffuant_mixing_2000,deffuant_2002,li_bounded_2020},  agents interact only if their opinions are within a certain threshold, and they adjust their opinions by averaging. These models typically use continuous opinion spaces and symmetric interactions.
\end{itemize}

The main novelty of our proposed opinion dynamics lies in the introduction of a distinctive mechanism whereby the probability of an agent updating their opinion is directly proportional to the opinion value of another agent. This feature implies that agents holding more extreme opinions exert greater influence on others, thereby fostering the emergence of phenomena such as opinion polarization and/or opinion clustering.

The remainder of this manuscript is organized as follows. Section~\ref{subsec:sec2.1} is devoted to the derivation of the mean-field formulation of our agent-based opinion dynamics in the large-population limit $N \to \infty$, leading to the coupled system of nonlinear Boltzmann-type ODEs \eqref{eqn:ODE_BA}. In section~\ref{subsec:sec2.2} we analyze the long-time behavior of the finite-agent model, keeping the number of agents $N$ fixed. Section~\ref{subsec:sec3.1} focuses on the mean-field dynamics in the special case where the opinion space consists of three opinion states, while section~\ref{subsec:sec3.2} extends the analysis to the general case of finitely many admissible opinions. Finally, section~\ref{sec:sec4} concludes the manuscript and presents several open and challenging problems for future research.

\section{The mean-field incremental voter model}\label{sec:sec2}
\setcounter{equation}{0}

\subsection{Derivation of the mean-field dynamical system}\label{subsec:sec2.1}
Throughout the rest of this manuscript, we assume that the collection of initial conditions $X_0^{1,N},\ldots,X_0^{N,N}$ are independent and ${\bf p}(0)$-distributed, where ${\bf p}(0)$ is a given probability mass function on $\Omega$. We also let ${\bf X}_t^N = (X_t^{1,N},\ldots,X_t^{N,N})$ to represent the state of the multi-agent system at time $t\geq 0$. Our goal is to derive a mean-field system that characterizes the behavior of a representative agent in our incremental voter model under the large-population limit $N \to \infty$. First of all, it is rather natural (via a heuristic mean-field argument by assuming the statistical independence of the pair of interacting agents) to expect, that the mean-field limit of the incremental voter model \eqref{eq:dynamics} as the number of agent tends to infinity should be given by

\begin{equation}\label{eqn:ODE_BA}
\left\{
\begin{aligned}
p'_{-k}(t) & = \frac{k-m(t)}{2k}\,p_{-k+1}(t) - \frac{k+m(t)}{2k}\,p_{-k}(t), \\
p'_n(t) & = \frac{k-m(t)}{2k}\,p_{n+1}(t) + \frac{k+m(t)}{2k}\,p_{n-1}(t) - p_n(t), ~~-k<n<k, \\
p'_k(t) & = \frac{k+m(t)}{2k}\,p_{k-1}(t) - \frac{k-m(t)}{2k}\,p_k(t),
\end{aligned}\right.
\end{equation}
where $t \geq 0$  and $m(t) \coloneqq \sum_{-k\leq n\leq k} n\,p_n(t)$. In the aforementioned mean-field description of the incremental voter model, the probability vector ${\bf p}(t) \coloneqq \left(p_{-k}(t),\ldots,p_k(t)\right)$ captures the distribution of opinions among an infinity large cloud of population and $p_n(t)$ represents the probability that a typical agent holds opinion $n$ at time $t$. With these interpretations, $m(t)$ represents the average opinion (in the mean-field regime) at time $t$, and $(k+m(t))/(2k)$ (respectively $(k-m(t))/(2k)$) denotes the average probability that a listener agent increments (respectively decrements) its opinion by one unit (if possible). The rigorous convergence of a finite multi-agent system to its mean-field limit under the large population limit $N \to \infty$ is commonly referred to as propagation of chaos \cite{sznitman_topics_1991}. This notion has been extensively studied in a broad range of disciplines, including social, economic, biological, and physical sciences \cite{cao_entropy_2021,cao_interacting_2022,cao_uncovering_2022,cao_uniform_2024,cao_wealth_2026,carlen_kinetic_2013,cortez_quantitative_2016,degond_macroscopic_2004}. In particular, it has played a central role in the derivation and analysis of kinetic models arising from socio-economic applications, especially those governed by Boltzmann-type and Fokker-Planck type equations \cite{cao_bennati_2025,during_boltzmann_2008,matthes_steady_2008}.

The rigorous derivation of the mean field system of nonlinear ODEs \eqref{eqn:ODE_BA} can be carried out using a probabilistic coupling approach, akin to those employed in several recent studies \cite{cao_derivation_2021,cao_fractal_2025}. For the ease of presentation, we first review some basic elements from probability theory. We adopt the notation $\mathcal{L}({\bf X})$ to denote the law of a generic random variable or vector ${\bf X}$. Also, we will quantify convergence of probability measures via the normalized Wasserstein distance (of order $1$): for probability mass functions ${\bf p}, {\bf q}$ on $\Omega^N$, it is defined by
\begin{equation}\label{def:Wasserstein}
\mathcal{W}_1({\bf p}, {\bf q}) \coloneqq \inf\limits_{{\bf X},{\bf Y}} \mathbb{E}\left[\frac{1}{N}\sum\limits_{i=1}^N |X^i - Y^i| \right],
\end{equation}
where the infimum is taken over all possible couplings of ${\bf p}$ and ${\bf q}$, or equivalently over all pair of random vectors ${\bf X} = (X^1,\ldots,X^N) \sim {\bf p}$ and ${\bf Y} = (Y^1,\ldots,Y^N) \sim {\bf q}$.

Now, we observe that the incremental voter model \eqref{eq:dynamics} can be described via the following system of SDEs driven by Poisson point measures:  \\

\begin{equation}
\label{eq:SDE_Xi}	
\begin{aligned}
\dd X_t^{i,N}
&= \sum_{j\neq i} \int_0^1 \left[\mathbbm{1}\left\{u \leq \frac{k + X_{t-}^{j,N}}{2k} \right\}\,\mathbbm{1}\left\{X_{t-}^{i,N} < k\right\} \right. \\
& \qquad \left. {} - \mathbbm{1}\left\{u >\frac{k + X_{t-}^{j,N}}{2k}\right\}\,\mathbbm{1}\left\{X_{t-}^{i,N} > -k\right\}
	\right] \, \cP^{ij}(\dd t, \dd u),
\end{aligned}
\end{equation}
where $\left\{\cP^{ij}(\dd t, \dd u)\right\}_{\substack{1\leq i,j\leq N\\ i\neq j}}$ is a collection of i.i.d. Poisson point measures on $[0,\infty) \times [0,1]$ with intensity $\frac{\dd t \, \dd u}{N-1}$.

For each $i$, we introduce a master clock for agent $i$ by setting
\begin{equation*}	
\cP^i(\dd t, \dd u) \coloneqq \sum\limits_{j\neq i} \cP^{ij}(\dd t, \dd u),
\end{equation*}
which is a Poisson point measure on $[0,\infty) \times [0,1]$ with intensity $\dd t \, \dd u$. Moreover, the collection $(\cP^i)_{1 \leq i \leq N}$ is again independent. Now we can write the following SDE for the evolution of the agent $i$'s opinion $X^{i,N}$, which is equivalent in law to \eqref{eq:SDE_Xi}:

\begin{equation}
\label{eq:equivalent_SDE_Xi}	
\begin{aligned}
\dd X_t^{i,N}
&= \int_0^1 \left[\mathbbm{1}\left\{u \leq \frac{k + \cA_{t-}^{i,N}}{2k} \right\}\,\mathbbm{1}\left\{X_{t-}^{i,N} < k\right\} \right. \\
& \qquad \left. {} - \mathbbm{1}\left\{u >\frac{k + \cA_{t-}^{i,N}}{2k}\right\}\,\mathbbm{1}\left\{X_{t-}^{i,N} > -k\right\}
	\right] \, \cP^i(\dd t, \dd u),
\end{aligned}
\end{equation}
in which \[\mathcal{A}^{i,N}_t \coloneqq \frac{1}{N-1} \sum\limits_{j \neq i} X^{j,N}_t\] represents the empirical average of the remaining $N−1$ agents aside from agent $i$ at time $t$. Based on the structure of the dynamics \eqref{eq:equivalent_SDE_Xi}, it is straightforward to expect that the mean-field limit process $(Z_t)_{t\geq 0}$ which captures the opinion dynamics of a typical agent in the asymptotic regime of an infinite population, is governed by the following nonlinear SDE:
\begin{equation}
\label{eq:mean_field_SDE_Zi}	
\begin{aligned}
\dd Z_t^i
&= \int_0^1 \left[\mathbbm{1}\left\{u \leq \frac{k + m_t}{2k} \right\}\,\mathbbm{1}\left\{Z_{t-}^i < k\right\} \right. \\
& \qquad \left. {} - \mathbbm{1}\left\{u >\frac{k + m_t}{2k}\right\}\,\mathbbm{1}\left\{Z_{t-}^i > -k\right\}
	\right] \, \cP^i(\dd t, \dd u),
\end{aligned}
\end{equation}
with $m_t = \EE[Z_t^i]$ for all $i$. We emphasize that components $(Z^i_t)_{1\leq i\leq N}$ of ${\bf Z}_t = (Z^1_t,\ldots,Z^N_t)$ are i.i.d. copies of the mean-field process \eqref{eq:mean_field_SDE} defined as the solution of the following SDE:
\begin{equation}
\label{eq:mean_field_SDE}	
\begin{aligned}
\dd Z_t
&= \int_0^1 \left[\mathbbm{1}\left\{u \leq \frac{k + \EE[Z_t]}{2k} \right\}\,\mathbbm{1}\left\{Z_{t-} < k\right\} \right. \\
& \qquad \left. {} - \mathbbm{1}\left\{u >\frac{k + \EE[Z_t]}{2k}\right\}\,\mathbbm{1}\left\{Z_{t-} > -k\right\}
	\right] \, \cP(\dd t, \dd u),
\end{aligned}
\end{equation}
where $\cP(\dd t, \dd u)$ is again a Poisson point measure on $[0,\infty) \times [0,1]$ with intensity $\dd t \, \dd u$. Specifically, each $Z^i$ is a solution of the mean-field SDE \eqref{eq:mean_field_SDE} with $\cP^i$ in place of $\cP$. To proceed, we employ a standard coupling method which is demonstrated in several recent works \cite{cao_derivation_2021,cao_fractal_2025}. The main strategy is to couple the multi-agent interacting system ${\bf X}^N_t = (X^{1,N}_t,\ldots,X^{N,N}_t)$ with the mean-field system ${\bf Z}_t = (Z^1_t,\ldots,Z^N_t)$ where the law of each component $Z^i_t$ is ${\bf p}(t)$ for $1\leq i\leq N$. Indeed, we use the same Poisson point measure $\cP^i(\dd t, \dd u)$ for the evolution of both $X^{i,N}_t$ and $Z^i_t$. For the initial conditions, we set $Z_0^i = X_0^{i,N}$ for all $1\leq i\leq N$ so that we have a synchronous coupling between the multi-agent system ${\bf X}^N_t$ and the system of independent mean-field processes ${\bf Z}_t$. In addition, the collection $(X_t^1, Z_t^1), \ldots (X_t^N, Z_t^N)$ is clearly exchangeable for all $t\geq 0$.

\begin{theorem}[Propagation of chaos]
\label{thm:PoC}
Assume that $X_0^{1,N}, \ldots, X_0^{N,N}$ are i.i.d. and distributed according to ${\bf p}(0) \in \cP(\Omega)$, then there exists a universal constant $C>0$ (independent of $t$) such that for all $t\geq 0$, it holds that
\begin{equation}\label{eq:PoC}
\mathcal{W}_1 \left(\mathcal{L}({\bf X}_t^N) \, , \, {\bf p}^{\otimes N}(t) \right) \leq \frac{C\,\expo^{t/k}}{\sqrt{N}}.
\end{equation}
\end{theorem}

\begin{proof}
The proof proceeds along similar lines to the argument given for Theorem 2.1 in the recent work \cite{cao_fractal_2025} (see also \cite{cao_fractal_2026}). However, for the sake of completeness we provide the full details here. For an arbitrary but fixed $i \in \{1,\ldots,N\}$, let $h(t) = \mathbb{E}[|X_t^{i,N} - Z_t^i|]$. Due to the aforementioned coupling, we automatically have
\[\mathcal{W}_1\left(\mathcal{L}({\bf X}_t^N) \, ,\, {\bf p}^{\otimes N}(t) \right) \leq \mathbb{E}\left[\frac{1}{N} \sum_{j=1}^N |X_t^{j,N} - Z_t^j| \right] = \mathbb{E}[|X_t^{i,N} - Z_t^i|].\] Consequently, it suffices to control $h(t)$ from above.

From the SDEs \eqref{eq:equivalent_SDE_Xi} and \eqref{eq:mean_field_SDE_Zi}, we observe that for every atom \((t, u)\) of \(\cP^i\) such that
\[u \in \left( \frac{k + \min\{ \cA_{t-}^{i,N},\, m_t \}}{2k},\, \frac{k + \max\{ \cA_{t-}^{i,N},\, m_t \}}{2k} \right),\]
the quantity \(|X_t^i - Z_t^i|\) increases (after a single jump) by at most \(2\). In all other scenarios, the quantity \(|X_t^i - Z_t^i|\) either decreases or remains unchanged. Consequently, we deduce that

\begin{equation*}
\begin{aligned}
\frac{\dd}{\dd t} h(t)
&\leq 2\,\EE \int_0^1 \mathbbm{1}\left\{\tfrac{k+\min\{\cA_{t-}^{i,N}, m_t\}}{2k} \leq u \leq \tfrac{k+\max\{\cA_{t-}^{i,N}, m_t\}}{2k} \right\}\, \dd u \\
&= \frac{1}{k}\, \EE \left|\mathcal{A}^{i,N}_t - m_t\right| \\
&\leq \frac{1}{k}\,\left(\EE \left| \mathcal{A}^{i,N}_t - \overbar{\mathcal{A}}^{i,N}_t \right| + \EE \left|\overbar{\mathcal{A}}^{i,N}_t - m_t \right| \right),
\end{aligned}
\end{equation*}
in which $\overbar{\mathcal{A}}^{i,N}_t \coloneqq \frac{1}{N-1} \sum\limits_{j\neq i} Z^j_t$. Thanks to the exchangeability of the agent-based dynamics, we have
\[\mathbb{E}\left|\mathcal{A}^{i,N}_t - \overbar{\mathcal{A}}^{i,N}_t\right| \leq h(t).\]
On the other hand, since $|Z_t^i| \leq k$ almost surely, the classical law of large numbers guarantees the existence of some finite constant $C > 0$ (independent of $t$) such that
\begin{equation*}
\EE \left|\overbar{\mathcal{A}}^{i,N}_t - m_t\right| \leq \frac{C}{\sqrt{N}}.
\end{equation*}
Putting these estimates together, we arrive at the following differential inequality:
\begin{equation*}
\frac{\dd}{\dd t} h(t) \leq \frac{1}{k}\,h(t) + \frac{C}{k\,\sqrt{N}}.
\end{equation*}
Finally, $h(0) = 0$ since initially we have ${\bf Z}_0 = {\bf X}^N_0$, from which the advertised bound \eqref{eq:PoC} follows thanks to Grownall's inequality.
\end{proof}


\subsection{Convergence to consensus for the finite system}\label{subsec:sec2.2}

We conclude this section by presenting results concerning the asymptotic behavior of the finite system, where the number of agents $N$ is fixed. For notational convenience, throughout this section we omit the superscript $N$ and write ${\bf X}_t$ in place of ${\bf X}^N_t$. First of all, we notice from the definition of agent-based model \eqref{eq:dynamics} that the two consensus/extreme states ${\bf X}^- \coloneqq (-k,\ldots,-k)$ and ${\bf X}^+ \coloneqq (k,\ldots,k)$ are absorbing. Since the Markov chain $\left({\bf X}_t\right)_{t\geq 0}$ is finite, its large time behavior must be absorption in a closed communicating class (consisting of an absorbing state or set of states) \cite{durrett_probability_2019,lanchier_stochastic_2017,liggett_interacting_1985}. Our main goal is to show that starting from any initial configuration ${\bf X}_0$ the chain is eventually absorbed in one of the two consensus states with probability $1$. For this purpose, we let $(\tau_\ell)_{\ell \geq 1}$ be a collection of i.i.d. $\textrm{Exp}(N/2)$ inter-arrival times of a Poisson clock running at rate $N/2$ and denote $T_n \coloneqq \sum_{\ell=1}^n \tau_\ell$ as the $n$-th ring time for each $n\geq 1$. Define the embedded discrete-time Markov chain by setting ${\bf Y}_0 \coloneqq {\bf X}_0$ and ${\bf Y}_n \coloneqq {\bf X}_{T_n}$ for $n\geq 1$, we can recover the continuous-time Markov chain using ${\bf X}_t = {\bf Y}_{N_t}$, where $N_t \coloneqq \max\{n \in \mathbb N \mid T_n \leq t\}$ represents the (random) number of rings up to and including time $t$.

\begin{proposition}\label{prop:conv_finite_chain}
For any fixed $N \geq 2$ and any initial configuration ${\bf X}_0$, the Markov chain $\left({\bf X}_t\right)_{t\geq 0}$ is absorbed in finite time almost surely in one of the two consensus states ${\bf X}^-$ or ${\bf X}^+$. In other words, we have
\begin{equation*}
\mathbb{P}\left({\bf X}_t \xrightarrow{t \to \infty} {\bf X}^- ~\textrm{or}~ {\bf X}_t \xrightarrow{t \to \infty} {\bf X}^+ \right) = 1.
\end{equation*}
\end{proposition}

\begin{proof}
We first prove the almost surely absorption of the embedded chain ${\bf Y}_n$. The crucial result on which we will rely is that for any non-consensus configuration ${\bf y} \in \Omega^N$, we can construct a finite sequence of ordered-pair selections and associated outcomes (all with positive conditional probabilities) which transform ${\bf y}$ into ${\bf X}^+$ in at most $L({\bf y}) \coloneqq \sum_{i=1}^N (k-y^i) \leq 2kN$ discrete steps. A similar statement holds with ${\bf X}^+$ replaced by ${\bf X}^-$.

Indeed, as ${\bf y} \notin \{{\bf X}^+,{\bf X}^-\}$, there exists at least one index $i_0$ with $y^{i_0} > -k$. Fix such an $i_0$ and we call agent $i_0$ the initial seed persuader. Our goal is to produce a finite list of ordered pairs $(i_n, j_n)_{1\leq n\leq L({\bf y})}$ and required outcomes such that applying those interactions in a sequel turns the opinion of every agent to $k$. The construction proceeds in two phases as follows:
\begin{itemize}
\item In the first phase, for every agent $i \neq i_0$, we perform the ordered pair interaction where agent $i$ is persuaded by $i_0$ repeatedly until agent $i$ reaches opinion $k$. Since the persuader is agent $i_0$ (whose opinion $y^{i_0}$ is larger than $-k$) and in this phase we never choose an ordered pair where $i_0$ is the listener agent, the opinion of agent $i_0$ remains unchanged during this phase and each desired increase outcome for a listener agent has conditional probability \[\frac 12 + \frac{y^{i_0}}{2k} \geq \frac 12 + \frac{-k+1}{2k} = \frac{1}{2k}.\] Thus for each single increase event while using agent $i_0$ as a persuader, there is a probability at least $p_{\min} = \frac{1}{2k} > 0$ of success conditional on that ordered pair being selected at that ring. Notice that the probability of selecting the particular ordered pair of agents $(i,i_0)$ at a ring time is $\frac{1}{N(N-1)} > 0$, thus in the first phase we can raise the opinion of every agent $i\neq i_0$ to opinion $k$ by performing $\sum_{i\neq i_0} \left(k - y^i\right)$ increase events in total (for each $i \neq i_0$ we issue that many $(i,i_0)$ persuade requests and require all those attempts to succeed). This is a finite prescribed sequence which has a positive probability bounded below by $\left(\frac{p_{\min}}{N(N-1)}\right)^{\sum_{i\neq i_0} \left(k - y^i\right)}$. After the completion of the first phase, the opinion of every agent except possibly agent $i_0$ will be at $k$, and agent $i_0$ still remains its original opinion $y^{i_0}$ which may be less than $k$.

\item In the second phase, we employ newly-created agents with opinion $k$ to raise the opinion of agent $i_0$. Since phase 1 has created at least one agent (in fact many) with opinion $k$, we choose any agent $j$ with opinion $k$ and repeatedly use the ordered pair $(i_0,j)$ requiring increase outcomes until agent $i_0$ reaches opinion $k$. As the persuader $j$ currently has opinion $k$, each required increase outcome occurs with probability $\frac 12 + \frac{k}{2k} = 1$, i.e., it will happen with probability 1 conditioned on selecting the ordered pair of agents $(i_0,j)$. Notice again that the probability to pick the desired ordered pair of agents $(i_0,j)$ at a ring time equals $\frac{1}{N(N-1)}$, thus the finite list of required interaction pairs $(i_0,j), (i_0,j), \ldots, (i_0,j)$ needed to raise the opinion of agent $i_0$ to $k$ occurs with probability at least $\left(\frac{1}{N(N-1)}\right)^{k-y^{i_0}} > 0$.
\end{itemize}
Concatenating the aforementioned two phases yields a finite explicit sequence of ordered pairs and desired increase outcomes with probability bounded below by \[\pi({\bf y}) \coloneqq \left(\frac{p_{\min}}{N(N-1)}\right)^{\sum_{i\neq i_0} \left(k - y^i\right)}\cdot \left(\frac{1}{N(N-1)}\right)^{k-y^{i_0}} > 0,\]
which transforms the initial opinion configuration ${\bf y}$ into ${\bf X}^+$. Similarly, for any ${\bf y} \notin \{{\bf X}^+,{\bf X}^-\}$ there is a positive probability to reach ${\bf X}^-$ within a finite step. Consequently, the only closed communicating classes are the singleton sets $\{{\bf X}^+\}$ and $\{{\bf X}^-\}$, and every other non-consensus opinion configuration is transient or leads with positive probability to one of those two absorbing states. From standard theory on finite Markov chains, this implies that the chain $({\bf Y}_n)_{n\geq 0}$ must be absorbed (i.e., enter one of the absorbing singletons) with probability 1. To be more precise, define the absorbtion time by \[\tau \coloneqq \inf\left\{n \in \mathbb N \mid {\bf Y}_n \in \{{\bf X}^+,{\bf X}^-\}\right\}.\] Since at each block of at most $2kN$ consecutive steps there is at least probability $p_0 \coloneqq \min_{{\bf y} \in \Omega^N} \pi(y) \geq (N\,(N-1))^{-2kN}\,(2\,k)^{-2kN} > 0$, that absorption occurs within that block, the absorption time $\tau$ is stochastically dominated by a geometric random variable with success probability $p_0$ (depending only on $N$ and $k$). Therefore, we deduce that $\mathbb{P}(\tau < \infty) = 1$ and $\mathbb{E}[\tau] \leq \frac{1}{p_0} < \infty$, whence $\tau$ is finite almost surely and the embedded discrete chain $({\bf Y}_n)_{n\geq 0}$ will be absorbed in finite discrete steps almost surely.

Since $\tau$ is a finite integer almost surely, the corresponding absorption time for the continuous-time Markov chain $({\bf X}_t)_{t\geq 0}$ defined by $S \coloneqq T_\tau \coloneqq \sum_{\ell=1}^\tau \tau_\ell$ is a sum of almost surely finite number of i.i.d. exponential random variables, hence $S$ is finite almost surely as well. Finally, for $t\geq S$ we have $N_t \geq \tau$, whence \[{\bf X}_t = {\bf Y}_{N_t} = {\bf Y}_\tau \in \{{\bf X}^+,{\bf X}^-\} \] for all $t\geq S$. This completes the proof of Proposition \ref{prop:conv_finite_chain}.
\end{proof}

\section{Large time analysis of the mean-field system}\label{sec:sec3}
\setcounter{equation}{0}

In this section, we undertake the analysis of the long-time behavior of solutions to the mean-field system of nonlinear ODEs \eqref{eqn:ODE_BA}. For the sake of convenience, we perform a harmless change of variable by setting
\begin{equation}\label{eq:relabling}
q_n \coloneqq p_{-k+n} \quad \text{for all $0 \leq n \leq 2k$}.
\end{equation}
Accordingly, we identify ${\bf p} \coloneqq (p_{-k},p_{-k+1},\ldots,p_{k-1},p_k)$ with ${\bf q} \coloneqq (q_0,q_1,\ldots,q_{2k-1},q_{2k})$. In essence, this amounts to shifting the opinion space $\Omega$ from $\{-k,\ldots,k\}$ to $\{0,\ldots,2k\}$ so that all admissible opinion values lie in $\mathbb N$. After this straightforward relabeling of the solution vector and corresponding shift of the opinion space, the ODE system \eqref{eqn:ODE_BA} takes the following equivalent form:

\begin{equation}\label{eqn:ODE_main}
\left\{
\begin{aligned}
q'_0(t) & = \frac{2k-M(t)}{2k}\,q_1(t) - \frac{M(t)}{2k}\,q_0(t), \\
q'_n(t)& = \frac{2k-M(t)}{2k}\,q_{n+1}(t) + \frac{M(t)}{2k}\,q_{n-1}(t) - q_n(t), ~~0<n<2k, \\
q'_{2k}(t) & = \frac{M(t)}{2k}\,q_{2k-1}(t) - \frac{2k-M(t)}{2k}\,q_{2k}(t),
\end{aligned}\right.
\end{equation}
in which $t \geq 0$ and $M(t) \coloneqq \sum_{0\leq n\leq 2k} n\,q_n(t)$. Next, we collect a few elementary observations on the solution of \eqref{eqn:ODE_main} in the following lemma.
\begin{lemma}\label{lem:1}
Assume that ${\bf q} = (q_0,\ldots,q_{2k})$ is a classical solution to the nonlinear system of ODEs \eqref{eqn:ODE_main} with ${\bf q}(t=0) \in \mathcal{P}(\{0,\ldots,2k\})$. Then we have ${\bf q}(t) \in \mathcal{P}(\{0,\ldots,2k\})$ for all $t \geq 0$. Moreover, the system \eqref{eqn:ODE_main} admits three equilibrium solutions given by ${\bf q}^l = (1,0,\ldots,0)$, ${\bf q}^u = \frac{1}{2k+1}(1,1,\ldots,1)$, and ${\bf q}^r = (0,\ldots,0,1)$, respectively. Equivalently, suppose that ${\bf p} = (p_{-k},\ldots,p_k)$ is a classical solution to the nonlinear system of ODEs \eqref{eqn:ODE_BA} with ${\bf p}(t=0) \in \mathcal{P}(\{-k,\ldots,k\})$. Then ${\bf p}(t) \in \mathcal{P}(\{-k,\ldots,k\})$ for all $t \geq 0$, and the system \eqref{eqn:ODE_BA} admits three equilibrium solutions given by ${\bf p}^l = (1,0,\ldots,0)$, ${\bf p}^u = \frac{1}{2k+1}(1,1,\ldots,1)$, and ${\bf p}^r = (0,\ldots,0,1)$, respectively.
\end{lemma}

\begin{proof}
For the sake of definiteness we will work with the original system \eqref{eqn:ODE_BA}. The fact that ${\bf p}(t=0) \in \mathcal{P}(\{-k,\ldots,k\})$ implies ${\bf p}(t) \in \mathcal{P}(\{-k,\ldots,k\})$ for all $t \geq 0$ is a simple consequence of the following conservation of probability mass:
\begin{align*}
\sum\limits_{n=-k}^k p'_n &= \frac{k-m}{2k}\,\sum\limits_{n=-k}^{k-1} p_{n+1} + \frac{k+m}{2k}\,\sum\limits_{n=-k+1}^k p_{n-1} \\
&\qquad - \sum\limits_{n=-k+1}^{k-1} p_n - \frac{k+m}{2k}\,p_{-k} - \frac{k-m}{2k}\,p_k \\
&= 0.
\end{align*}
At equilibrium where $p'_n = 0$ for all $-k\leq n\leq k$, we require that
\begin{equation}\label{eq:cond}
(k-m)\,p_{n+1} = (k+m)\,p_n ~~\textrm{for each $-k\leq n\leq k-1$}.
\end{equation}
A direct computation shows that the consensus opinion profile concentrated at either $-k$ (given by ${\bf p}^l$) or $k$ (given by ${\bf p}^r$) satisfies the condition \eqref{eq:cond}. Excluding these equilibrium distributions, we may assume without loss of generality that the average opinion $m = \sum_{-k\leq n\leq k} n\,p_n < k$ lies in the open interval $(-k,k)$. Since $-k < m < k$, we deduce from \eqref{eq:cond} that $\frac{p_{n+1}}{p_n} = \frac{k+m}{k-m} \coloneqq r \in (0,\infty)$ for all $-k\leq n\leq k-1$. In other words, ${\bf p}$ is a geometric distribution starting from $p_{-k}$ with the common ratio $r$. As ${\bf p} \in \mathcal{P}(\{-k,\ldots,k\})$, $p_{-k} = \frac{1-r}{1-r^{2k+1}}$. Employing the identity \[\frac{r-1}{r+1}\,k = m = \sum_{-k\leq n\leq k} n\,p_n = p_{-k}\,r^k\,\sum_{-k\leq n\leq k} n\,r^n,\] we end up with the following algebraic equation of degree $2k+1$ in $r > 0$:
\begin{equation}\label{eq:polynomial}
f(r) \coloneqq (2k-1)\,r^{2k+1} - (2k+1)\,r^{2k} + (2k+1)\,r - (2k-1) = 0.
\end{equation}
An elementary investigation on the derivative of $f$ shows that $f(r)$ is strictly increasing for all $r > 0$ and $f(1) = 0$, so that \eqref{eq:polynomial} admits a unique positive root $r = 1$, from which we conclude that the uniform distribution ${\bf p}^u$ is likewise an equilibrium of the ODE system \eqref{eqn:ODE_BA}.
\end{proof}

\subsection{Quantitative convergence for three opinion states}\label{subsec:sec3.1}

Now we turn to the analysis of the large time behavior of the (equivalent) mean-field incremental voter model \eqref{eqn:ODE_main}. We first consider the special case $k=1$, in which the following convergence guarantee applies.

\begin{proposition}\label{prop:conv_k=1}
Assume that ${\bf q} = (q_0,q_1,q_2)$ is a classical solution to the nonlinear system of ODEs \eqref{eqn:ODE_main} with $k=1$ and ${\bf q}(t=0) \in \mathcal{P}(\{0,1,2\})$. Then we have the following three distinct scenarios (depending on the size of the initial average opinion):
\begin{enumerate}[label=(\roman*)]
\item If $M(0) > 1$, then ${\bf q}(t) \xrightarrow{t \to \infty} {\bf q}^r$ component-wisely.
\item If $M(0) < 1$, then ${\bf q}(t) \xrightarrow{t \to \infty} {\bf q}^l$ component-wisely.
\item If $M(0) = 1$, then ${\bf q}(t) \xrightarrow{t \to \infty} {\bf q}^u$ and the convergence of each component is exponentially fast.
\end{enumerate}
\end{proposition}

\begin{proof}
We observe that the mean opinion $M = q_1 + 2\,q_2$ evolves according to
\begin{equation}\label{eq:M_derivative}
M' = q'_1 + 2\,q'_2 = (1-M/2)\,q_2 + (M/2)\,q_0 - q_1 + M\,q_1 - 2\,q_2 + M\,q_2 = \frac{(M-1)\,q_1}{2}.
\end{equation}
Thus if we assume that $M(0) > 1$, then $M$ is non-decreasing and $M(t) > 1$ for all $t\geq 0$. Since $M$ is evidently bounded above by $2$, we deduce the existence of $M_\infty \coloneqq \lim_{t\to \infty} M(t)$. As $M$ is uniformly bounded with a uniformly bounded derivative, the aforementioned convergence of $M$ as $t \to \infty$ also guarantees that $\lim_{t \to \infty} M'(t) = 0$. This observation together with the strict inequality $M(t) > 1$ for all $t\geq 0$ yields the convergence $q_1(t) \xrightarrow{t \to \infty} 0$, which in turn implies the convergence of $q_2(t)$ in the large time limit (since both limits $\lim_{t\to \infty} M(t)$ and $\lim_{t\to \infty} q_1(t)$ exist). Due to the conservation of the total probability mass, we deduce the large time convergence of $q_0(t)$ as well, whence the element-wise limit $\lim_{t \to \infty} {\bf q}(t)$ of the solution vector ${\bf q}(t)$ exists. The uniform boundedness of the derivatives of $q_n$ for each $n \in \{0,1,2\}$ then ensures that ${\bf q}(t)$ converges to an equilibrium solution of the ODE system \eqref{eqn:ODE_main} with $k=1$. In view of the first-moment condition $\inf_{t\geq 0} M(t) > 1$, we then conclude that $\lim_{t \to \infty} {\bf q}(t) = {\bf q}^r$. This completes the proof of (i). By analogous reasoning, the second statement (ii) follows immediately. It remains to establish the last statement, where we assume that $M(0) =1$. In this scenario we have $M(t) = 1$ (or equivalently $q_2(t) = q_0(t)$) for all $t\geq 0$, so that the nonlinear system of ODEs \eqref{eqn:ODE_main} (with $k=1$) boils down to the following linear system:
\begin{equation}\label{eq:linear_ODE}
\frac{\dd}{\dd t}\begin{pmatrix} q_0 \\ q_1 \\ q_2 \end{pmatrix} = \begin{pmatrix} -1/2 & 1/2 & 0 \\ 1/2 & -1 & 1/2 \\ 0 & 1/2 & -1/2\end{pmatrix}\begin{pmatrix} q_0 \\ q_1 \\ q_2 \end{pmatrix}.
\end{equation}
The linear ODE system \eqref{eq:linear_ODE} can be solved explicitly, leading us to \[{\bf q}(t) = \left(q_0(t), q_1(t), q_2(t)\right) = {\bf q}^u + \left(q_0(0)- 1/3\right)\expo^{-\tfrac{3}{2}\,t}\,(1, -2, 1).\] As a result, the proof of (iii) is also completed.
\end{proof}

We present a series of numerical experiments illustrating the large-time behavior of the opinion profile $\mathbf{q}(t)$ and the evolution of the associated average opinion $M(t)$. Throughout, we take $k = 1$ and solve the ODE system \eqref{eqn:ODE_main} using the standard fourth-order Runge--Kutta method with time step $\Delta t = 0.01$. In Figure~\ref{fig:numerics}-left and Figure~\ref{fig:numerics}-right, we plot, respectively, the evolution of the opinion profile $\mathbf{q}(t)$ and the corresponding trajectory of $M(t)$. The initial data used in the top, middle, and bottom panels are $\mathbf{q}(0) = (0.4, 0.3, 0.3)$, $\mathbf{q}(0) = (0.3, 0.4, 0.3)$, and $\mathbf{q}(0) = (0.3, 0.3, 0.4)$, which yield initial average opinions $M(0) = 0.9 < 1$, $M(0) = 1$, and $M(0) = 1.1 > 1$, respectively. The numerical outcomes align well with the convergence behavior predicted by Proposition \ref{prop:conv_k=1}.

\def\figscale{0.5}
\def\figwidth{0.47\textwidth}
\begin{figure}
	\begin{subfigure}{\figwidth}
		\centering
		\includegraphics[scale=\figscale]{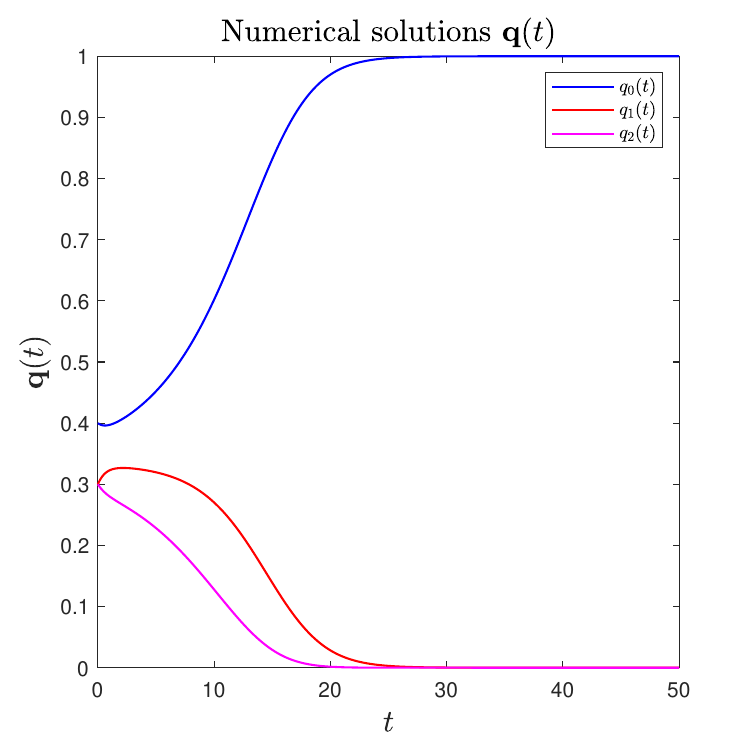}
	\end{subfigure}
	\hspace{0.1in}
	\begin{subfigure}{\figwidth}
		\centering
		\includegraphics[scale=\figscale]{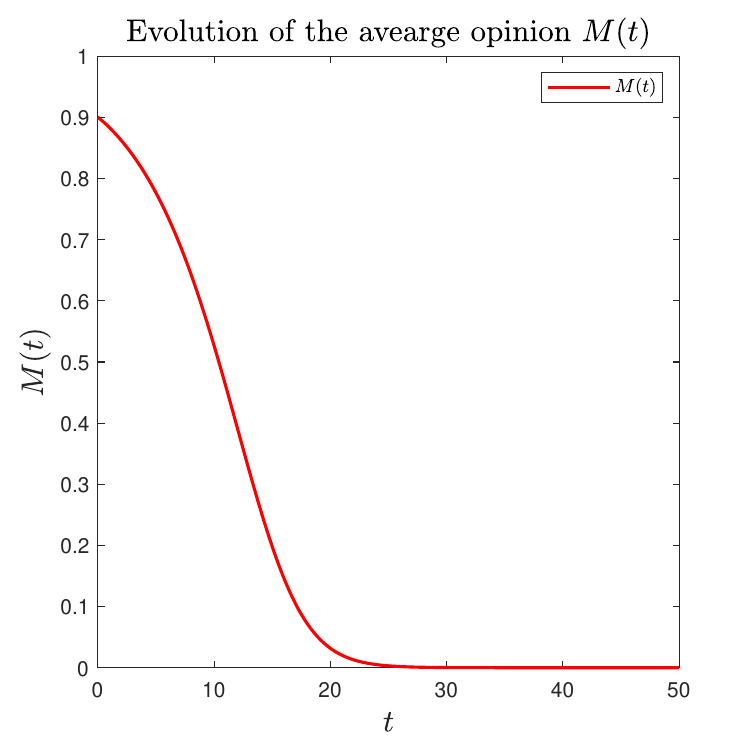}
	\end{subfigure}
	\\
	\begin{subfigure}{\figwidth}
		\centering
		\includegraphics[scale=\figscale]{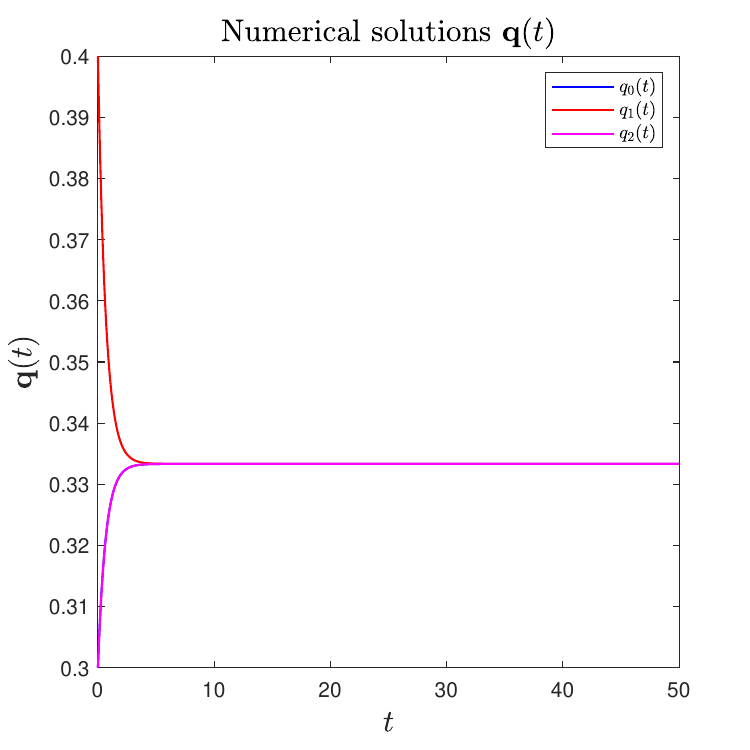}
	\end{subfigure}
	\hspace{0.1in}
	\begin{subfigure}{0.45\textwidth}
		\centering
		\includegraphics[scale=\figscale]{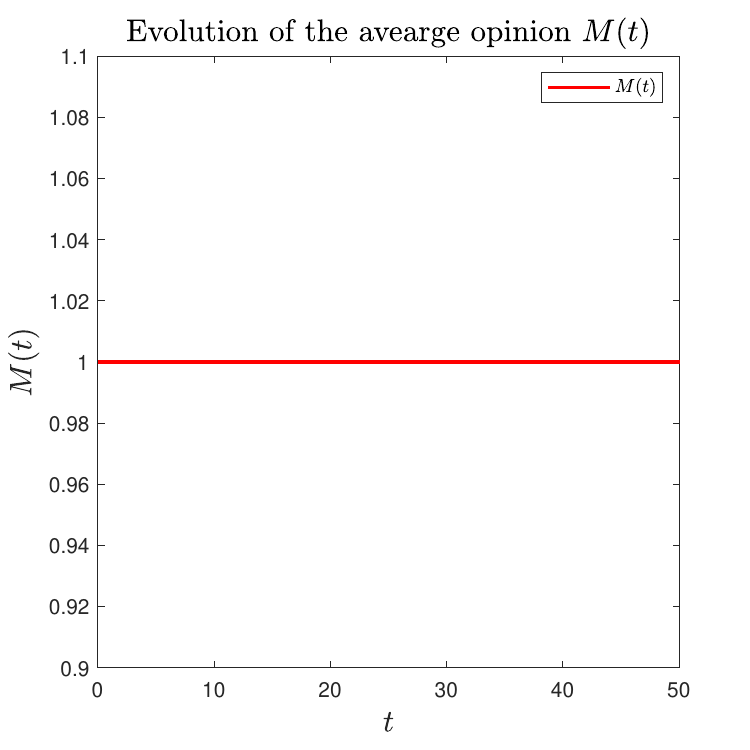}
	\end{subfigure}
    \\
    \begin{subfigure}{\figwidth}
    	\centering
    	\includegraphics[scale=\figscale]{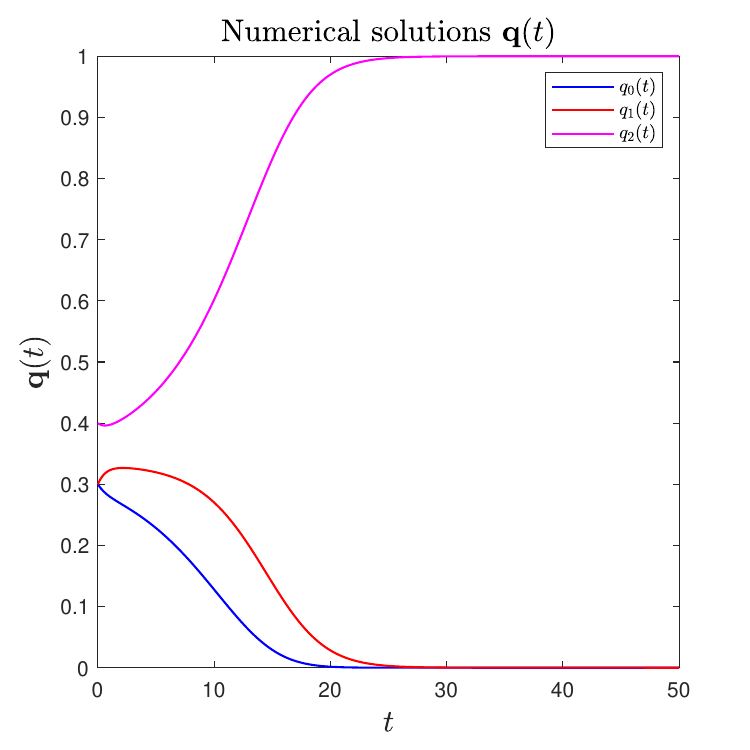}
    \end{subfigure}
    \hspace{0.1in}
    \begin{subfigure}{\figwidth}
    	\centering
    	\includegraphics[scale=\figscale]{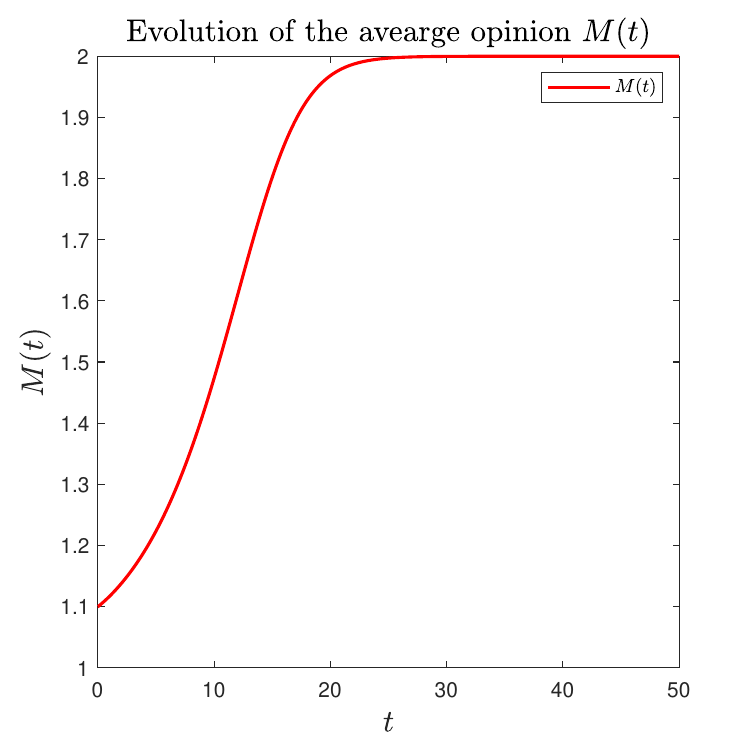}
    \end{subfigure}
	
\caption{{\bf Left}: Simulations of the mean-field dynamical system \eqref{eqn:ODE_main} with $k=1$, for three different initial datum with $M(0) < 1$ (top), $M(0) = 1$ (middle), and $M(0) > 1$ (bottom). {\bf Right}: Evolution of the corresponding average opinion with respect to time.}
\label{fig:numerics}
\end{figure}

Building on Proposition \ref{prop:conv_k=1}, we can strengthen the qualitative convergence results stated in parts (i) and (ii) into quantitative convergence guarantees.

\begin{theorem}\label{thm:conv_k=1}
Assume that ${\bf q} = (q_0,q_1,q_2)$ is a classical solution to the nonlinear system of ODEs \eqref{eqn:ODE_main} with $k=1$ and ${\bf q}(t=0) \in \mathcal{P}(\{0,1,2\})$. Let
\begin{equation*}
\Delta(t) \coloneqq \min\left\{2-M(t), M(t)\right\} = \begin{cases}
2 - M(t), & \textrm{if~ $M(t) > 1$}, \\
M(t),     & \textrm{if~ $M(t) < 1$}.
\end{cases}
\end{equation*}
Then there exists some universal constant $C$ (depending only on ${\bf q}(0)$) such that
\begin{equation}\label{eq:bound_on_Delta}
\Delta(t) \leq C\,\expo^{-\frac{t}{8}} ~~\textrm{for all $t \geq 0$}.
\end{equation}
In particular, for all $t\geq 0$ we have $\max\{q_0(t), q_1(t)\} \leq C\,\expo^{-\frac{t}{8}}$ if $M(0) > 1$, and \\
$\max\{q_1(t), q_2(t)\} \leq C\,\expo^{-\frac{t}{8}}$ if $M(0) < 1$.
\end{theorem}

\begin{proof}
To start with, we assume that $M(0) > 1$ so that $\Delta(0) = 2 - M(0)$. Thanks to the analysis carried out in Proposition \ref{prop:conv_k=1}, we have $\Delta(t) \in [0,1)$ for all $t\geq 0$. Next, one can check via a short algebraic simplification that \[\Delta' = -M' = -\frac{(M-1)\,q_1}{2} = -\frac{(1-\Delta)\,q_1}{2}.\] Ideally, we would like to derive a useful closed first-order differential inequality in $\Delta$ alone that holds globally, but this is probably not possible since $q_1$ might be arbitrarily small compared to $\Delta$. To overcome the aforementioned issue, notice that $q_2 = (2-\Delta - q_1) / 2$ and $q_0 = (\Delta - q_1) / 2$, thus we arrive at the following two-dimensional autonomous system
\begin{equation}\label{eq:2D}
\begin{cases}
\Delta' &= -\frac{q_1}{2}\,(1-\Delta), \\
q'_1 &=  -\frac 12\,\Delta^2 + \Delta - \frac 32\,q_1.
\end{cases}
\end{equation}
Our objective is to analyze the ratio dynamics of $r \coloneqq \frac{q_1}{\Delta}$ and establish a time-uniform lower bound on $r$ after some finite time, which will in turn leads us to a closed differential inequality in $\Delta$ valid beyond that time. Indeed, since
\[r' = \frac{\Delta\,q'_1 - q_1\,\Delta'}{\Delta^2} = 1 - \frac 32\,r + \frac 12\,r^2 - \frac{\Delta}{2}\,(1+r^2) > \frac{1-3r}{2},\] in which the last inequality follows readily from the fact that $\Delta(t) \in [0,1)$ for all $t\geq 0$, hence $r$ will increase if $r < \frac 13$. In fact, the differential inequality in $r$ guarantees that $r(t) \geq \frac 13 + \left(r(0) - \frac 13\right)\expo^{-\frac{3}{2}t}$, whence there exists some finite $T > 0$ (depending only on $r(0)$) for which $r(t) \geq \frac 14$ whenever $t \geq T$. Assembling the first identity in \eqref{eq:2D} and the lower bound $q_1 \geq \frac 14\,\Delta$ valid for $t\geq T$, we deduce that $\Delta' \leq -\frac 18\,\Delta\,(1-\Delta)$ for all $t \geq T$. Consequently, \[2\,q_0(t) + q_1(t) = 2 - M(t) = \Delta(t) \leq \frac{\Delta(T)}{\Delta(T)+\left(1-\Delta(T)\right)\expo^{(t-T)/8}} ~~\textrm{for all $t \geq T$},\] from which the advertised bound \eqref{eq:bound_on_Delta} follows readily.

The case when $M(0) < 1$ can be handled in a rather similar way, once we notice the following invariance of the ODE system
\begin{equation}\label{eq:ODE_with_k=1}
\begin{cases}
q'_0 &= q_1\,\left(1 - \frac{M}{2}\right) - q_0\,\frac{M}{2}, \\
q'_1 &= q_2\,\left(1 - \frac{M}{2}\right) + q_0\,\frac{M}{2} - q_1, \\
q'_2 &= q_1\,\frac{M}{2} - q_2\,\left(1 - \frac{M}{2}\right).
\end{cases}
\end{equation}
under a state-flip transformation prescribed by $\hat{q}_0 \coloneqq q_2$, $\hat{q}_1 \coloneqq q_1$, and $\hat{q}_2 \coloneqq q_0$. Indeed, for the distribution $\hat{q}$, the mean value is given by $\widehat{M} \coloneqq \hat{q}_1 + 2\,\hat{q}_2 = q_1 + 2\,q_0 = 2 - M$. If we replace $(q_0,q_1,q_2, M)$ by $(\hat{q}_2,\hat{q}_1,\hat{q}_0,2-\widehat{M})$, after a direct substitution and some algebraic simplifications we obtain exactly the same-form system for $\hat{q}$ but with mean $\widehat{M}$:
\begin{equation}\label{eq:ODE_with_k=1_rewrite}
\begin{cases}
\hat{q}'_0 &= \hat{q}_1\,\left(1 - \frac{\widehat{M}}{2}\right) - \hat{q}_0\,\frac{\widehat{M}}{2}, \\
\hat{q}'_1 &= \hat{q}_2\,\left(1 - \frac{\widehat{M}}{2}\right) + \hat{q}_0\,\frac{\widehat{M}}{2} - \hat{q}_1, \\
\hat{q}'_2 &= \hat{q}_1\,\frac{\widehat{M}}{2} - \hat{q}_2\,\left(1 - \frac{\widehat{M}}{2}\right).
\end{cases}
\end{equation}
Thus the ODE system \eqref{eq:ODE_with_k=1} is invariant under the flip ${\bf q} \mapsto \hat{{\bf q}}$. As a result, if $M(0) < 1$ then $\widehat{M}(0) > 1$, whence all arguments used in the case $M(0) > 1$ carry over to the probability distribution $\hat{{\bf q}}$ with the same constants. Mapping back to the original variables then provides the corresponding statement for ${\bf q}$ (with the same quantitative rates of convergence).
\end{proof}

\subsection{Large-time convergence for finitely many opinion states}\label{subsec:sec3.2}

We devote the remainder of this section to analyzing the long-time behavior of the mean-field opinion dynamics \eqref{eqn:ODE_main} for general $k \geq 2$. It is worth emphasizing that when $k = 1$, the mean opinion $M(t)$ interacts with only a single additional degree of freedom given by the central mass $q_1$. Consequently, the system \eqref{eqn:ODE_main} with $k = 1$ effectively reduces to a two-dimensional coupled ODE system, and the asymptotic behavior of the full solution ${\bf q}(t)$ can be inferred solely from the evolution of the mean opinion $M(t)$. In contrast, for a general $k \geq 2$, the dynamics of $M(t)$ no longer provide sufficient information to characterize the evolution of ${\bf q}(t)$. The analysis therefore becomes qualitatively more intricate and technically demanding in the general case when $k\geq 2$.


First of all, we investigate the simplest situation where the initial datum ${\bf q}(0)$ enjoys a precise symmetry property, described as follows:
\begin{definition}
We say that a probability mass function ${\bf q} \in \mathcal{P}(\Omega)$ is symmetric if $q_n = q_{2k-n}$ for all $0\leq n \leq 2k$. In addition, we write ${\bf q} \in \mathcal{S}$ if and only if ${\bf q} \in \mathcal{P}(\Omega)$ is symmetric.
\end{definition}

\begin{proposition}\label{prop:symmetric_IC}
Assume that ${\bf q} = (q_0,\ldots,q_{2k})$ is a classical solution to the nonlinear system of ODEs \eqref{eqn:ODE_main} with ${\bf q}(t=0) \in \mathcal{S}$. Then for all $t \geq 0$ we have \[\|{\bf q}(t) - {\bf q}^u\|_{\ell^2} \leq \expo^{-\frac{t}{8\,(2k+1)^2}} \|{\bf q}(0) - {\bf q}^u\|_{\ell^2}.\]
\end{proposition}

\begin{proof}
For the ease of notation, we identify ${\bf q}$ with its transposed version ${\bf q}^\intercal$. We notice that the symmetric subspace $\mathcal{S} \subset \mathcal{P}(\Omega)$ is invariant under the evolution of the mean-field ODE system \eqref{eqn:ODE_main}. In other words, if $q_n(0) = q_{2k-n}(0)$ for all $0\leq n\leq 2k$, then the solution ${\bf q}(t)$ satisfies $q_n(t) = q_{2k-n}(t)$ for all $t\geq 0$ and all $0\leq n\leq 2k$. In particular, the system boils down to the following linear system:
\begin{equation}\label{eq:linear_ODE_general}
\frac{\dd}{\dd t} {\bf q} = L\,{\bf q},
\end{equation}
where
\[L \coloneqq \begin{pmatrix}
-1/2 & 1/2 & 0 & 0 & \cdots & 0 & 0 & 0 \\
1/2 & -1 & 1/2 & 0 & \cdots & 0 & 0 & 0 \\
0 & 1/2 & -1 & 1/2 & \cdots & 0 & 0 & 0 \\
\vdots & \ddots & \ddots & \ddots & \ddots & \vdots & \vdots & \vdots \\
0 & 0 & 0 & 0 & \cdots & 1/2 & -1 & 1/2 \\
0 & 0 & 0 & 0 & \cdots & 0 & 1/2 & -1/2
\end{pmatrix}.\]
The linear system admits a unique solution given by ${\bf q}(t) = \expo^{t\,L}\,{\bf q}(0)$, and we can easily verify that the uniform distribution ${\bf q}^u$ is an equilibrium associated with the linear dynamics \eqref{eq:linear_ODE_general}. In order to prove the claimed exponential convergence towards ${\bf q}^u$, it suffices to investigate the spectral gap of the matrix $L$.

The matrix $L$ has a simple zero eigenvalue associated with the (normalized) eigenvector ${\bf q}^u$, thus its spectral gap will be provided by the second eigenvalue. Thanks to the content of Lemma \ref{lem:auxillary} below, we have
\begin{equation}\label{eq:spectral_gap}
\langle -L\left({\bf q} - {\bf q}^u\right), {\bf q} - {\bf q}^u \rangle \geq \frac{1}{8\,(2k+1)^2}\,\|{\bf q} - {\bf q}^u\|^2_{\ell^2}
\end{equation}
for all ${\bf q} \in \mathcal{P}(\Omega)$. Based on the aforementioned control on the spectral gap of $L$, we deduce that
\begin{equation*}
\begin{aligned}
\frac 12\,\frac{\dd}{\dd t} \|{\bf q}(t) - {\bf q}^u\|^2_{\ell^2} &= \langle L\,{\bf q}(t), {\bf q}(t) - {\bf q}^u \rangle \\
&= \langle L\left({\bf q}(t) - {\bf q}^u\right), {\bf q}(t) - {\bf q}^u \rangle \\
&\leq -\frac{1}{8\,(2k+1)^2}\,\|{\bf q}(t) - {\bf q}^u\|^2_{\ell^2},
\end{aligned}
\end{equation*}
whence the proof of Proposition \ref{prop:symmetric_IC} is completed after a routine application of Gr\"onwall's inequality.
\end{proof}

\begin{lemma}\label{lem:auxillary}
Let ${\bf x} = (x_0,x_1,\ldots, x_{2k}) \in \mathbb{R}^{2k+1}$ satisfies $\sum_{n=0}^{2k} x_n = 0$. Then we have
\[\langle -L\,{\bf x}, {\bf x} \rangle \geq \frac{1}{8\,(2k+1)^2}\,\|{\bf x}\|^2_{\ell^2}.\]
\end{lemma}

\begin{proof}
A straightforward computation shows that \[\langle -L\,{\bf x}, {\bf x} \rangle = \frac 12\,\sum_{n=0}^{2k-1} |x_{n+1}-x_n|^2.\] Since $\sum_{n=0}^{2k} x_n = 0$, we can write each coordinate $x_n$ as a telescoping sum as follows:
\begin{align*}
x_n &= x_0 + \sum_{m=0}^{n-1} (x_{m+1} - x_m) - \frac{1}{2k+1}\,\sum_{m=0}^{2k} x_m  \\
&= x_0 + \sum_{m=0}^{n-1} (x_{m+1} - x_m) - \frac{1}{2k+1}\,\sum_{m=0}^{2k} \left(x_0 + \sum_{r=0}^{m-1} (x_{r+1} - x_r) \right) \\
&= \sum_{m=0}^{n-1} (x_{m+1} - x_m) -  \frac{1}{2k+1}\,\sum_{m=0}^{2k}\sum_{r=0}^{m-1} (x_{r+1} - x_r).
\end{align*}
Consequently, the triangle inequality gives rise to
\begin{align*}
|x_n| &\leq \sum_{m=0}^{n-1} |x_{m+1} - x_m| + \frac{1}{2k+1}\,\sum_{m=0}^{2k}\sum_{r=0}^{m-1} |x_{r+1} - x_r| \\
&\leq \sum_{m=0}^{2k-1} |x_{m+1} - x_m| + \frac{1}{2k+1}\,\sum_{m=0}^{2k}\sum_{r=0}^{2k-1} |x_{r+1} - x_r| \\
&= 2\,\sum_{m=0}^{2k-1} |x_{m+1} - x_m|.
\end{align*}
Squaring and summing over $0\leq n\leq 2k$ yields
\[\sum_{n=0}^{2k} |x_n|^2 \leq 4\,(2k+1)\,\left(\sum_{n=0}^{2k-1} |x_{n+1} - x_n|\right)^2 \leq 4\,(2k+1)^2\,\sum_{n=0}^{2k-1} |x_{n+1} - x_n|^2,\]
from which the advertised estimate follows.
\end{proof}


For a generic and non-symmetric initial opinion profile ${\bf q}(0) \in \cP(\Omega) \setminus \cS$, one naturally expects that the solution ${\bf q}(t)$ of the mean-field dynamical system \eqref{eqn:ODE_main} converges to either the consensus state ${\bf q}^l$ at the leftmost opinion or the consensus state ${\bf q}^u$ at the rightmost opinion. We denote by $\cA^l$ and $\cA^r$ the basins of attraction of the (stable) equilibrium opinion profiles ${\bf q}^l$ and ${\bf q}^r$, respectively. In other words, we define
\begin{equation}\label{eq:Al}
\cA^l \coloneqq \left\{{\bf q}(0) \in \cP(\Omega) \mid \lim_{t\to \infty} {\bf q}(t) = {\bf q}^l \right\}~~\textrm{and}~~ \cA^r \coloneqq \left\{{\bf q}(0) \in \cP(\Omega) \mid \lim_{t\to \infty} {\bf q}(t) = {\bf q}^r \right\}.
\end{equation}
Due to the high-dimensionality of the nonlinear system \eqref{eqn:ODE_main} when $k\geq 2$, we believe that it is no longer feasible to obtain a complete characterization of $\cA^l$ or $\cA^r$ anymore. However, we managed to identify certain explicit subsets of these basins of attraction, which is summarized below:

\begin{theorem}\label{thm:basins_of_attraction}
Assume that ${\bf q} = (q_0,\ldots,q_{2k})$ is a classical solution to the nonlinear system of ODEs \eqref{eqn:ODE_main} with ${\bf q}(t=0) \in \cP(\Omega)$. Then we have the following convergence guarantees:
\begin{enumerate}[label=(\roman*)]
\item If $M(0) > 2k-1$, then ${\bf q}(t) \xrightarrow{t \to \infty} {\bf q}^r$ component-wisely.
\item If $M(0) < 1$, then ${\bf q}(t) \xrightarrow{t \to \infty} {\bf q}^l$ component-wisely.
\end{enumerate}
\end{theorem}

\begin{proof}
The key quantity we shall investigate is again the average opinion $M = \sum_{n=0}^{2k} n\,q_n$, whose temporal evolution is dictated by
\begin{equation*}
\begin{aligned}
M' &= \frac{2k-M}{2k}\,\sum_{n=1}^{2k-1} n\,q_{n+1} + \frac{M}{2k}\,\sum_{n=1}^{2k-1} n\,q_{n-1} - \sum_{n=1}^{2k-1} n\,q_n + M\,q_{2k-1} - (2k-M)\,q_{2k} \\
&= \frac{2k-M}{2k}\left(M-q_1-(1-q_0-q_1)\right) + \frac{M}{2k}\left(M-(2k-1)\,q_{2k-1}-2k\,q_{2k}+1-q_{2k-1}-q_{2k}\right) \\
&\qquad - (M-2k\,q_{2k}) +
M\,q_{2k-1} - (2k-M)\,q_{2k} \\
&= \frac{M}{k} - 1 + \frac{2k-M}{2k}\,q_0 - \frac{M}{2k}\,q_{2k}.
\end{aligned}
\end{equation*}
We observe that
\begin{equation*}
\begin{aligned}
\frac{2k-M}{2k}\,q_0 - \frac{M}{2k}\,q_{2k} &= \frac{2k-M}{2k}\,\frac{2k-M-\sum_{j=1}^{2k-1} (2k-j)\,q_j}{2k} - \frac{M}{2k}\,\frac{M-\sum_{j=1}^{2k-1} j\,q_j}{2k} \\
&= \frac{(2k-M)^2 - (2k-M)\,\sum_{j=1}^{2k-1} (2k-j)\,q_j - M^2 + M\,\sum_{j=1}^{2k-1} j\,q_j}{4k^2} \\
&= 1 - \frac{M}{k} + \frac{1}{4k^2}\,\sum_{j=1}^{2k-1} \left(2\,k\,(j+M) - 4\,k^2\right)\,q_j.
\end{aligned}
\end{equation*}
Consequently, we arrive at the following expression:
\begin{equation}\label{eq:M_derivative}
M' = \frac{1}{2\,k^2}\,\sum_{j=1}^{2k-1} \left(k\,(j+M) - 2\,k^2\right)\,q_j.
\end{equation}
Bearing in mind the trivial bound $0\leq M \leq 2\,k$ which is valid for all times, we deduce that if at some (finite) time $t\geq 0$ it holds that $M(t) > 2k-1$, then
\begin{equation}\label{eq:I1}
\begin{aligned}
M'(t) &= \frac{1}{2\,k^2}\sum_{j=1}^{2k-1} \left(k\,(j+M(t)) - 2\,k^2\right)\,q_j(t) \\
&\geq \frac{1}{2\,k^2}\sum_{j=1}^{2k-1} \left(k\,(1+2k-1)-2\,k^2\right)\,q_j(t) = 0.
\end{aligned}
\end{equation}
Therefore, if ${\bf q}(t=0) \in \cP(\Omega)$ is chosen such that $M(0) > 2k-1$, $M(t)$ will be non-decreasing in $t$ and hence converges to some $M_\infty \coloneqq \lim_{t \to \infty} M(t) \in (2k-1,2k]$. The uniform boundedness of $M(t)$ together with the uniform boundedness of $M'(t)$ implies that $\lim_{t\to\infty} M'(t) = 0$. Consequently, $\lim_{t \to \infty} q_j(t) = 0$ for each $1\leq j \leq 2k-1$.
From this, we further deduce that $\lim_{t\to\infty} q_0(t) = 0$ and $\lim_{t\to\infty} q_{2k}(t) = 1$. In other words, the full state vector satisfies $\lim_{t\to\infty} \mathbf{q}(t) = \mathbf{q}^r$ whenever $M(0) > 2k - 1$. Similarly, if at some (finite) time $t\geq 0$ it holds that $M(t) < 1$, then
\begin{equation}\label{eq:I2}
\begin{aligned}
M'(t) &= \frac{1}{2\,k^2}\sum_{j=1}^{2k-1} \left(k\,(j+M(t)) - 2\,k^2\right)\,q_j(t) \\
&\leq \frac{1}{2\,k^2}\sum_{j=1}^{2k-1} \left(k\,(2k-1+1)-2\,k^2\right)\,q_j(t) = 0.
\end{aligned}
\end{equation}
As a result, if the initial condition ${\bf q}(0) \in \mathcal{P}(\Omega)$ satisfies $M(0) < 1$, then $M(t)$ is non-increasing in time and therefore converges to a limit $M_\infty \coloneqq \lim_{t\to\infty} M(t) \in [0,1)$. By an argument entirely parallel to the one used above, we conclude that $\lim_{t \to \infty} \mathbf{q}(t) = \mathbf{q}^l$ whenever $M(0) < 1$.
\end{proof}

We illustrate the convergence results established in Proposition~\ref{prop:symmetric_IC} and Theorem~\ref{thm:basins_of_attraction} through a series of numerical experiments. Throughout, we set $k = 2$ and solve the ODE system \eqref{eqn:ODE_main} using again the standard fourth-order Runge--Kutta method with time step $\Delta t = 0.01$. In Figure~\ref{fig:numerics2}-left and Figure~\ref{fig:numerics2}-right, we plot respectively, the evolution of the opinion profile $\mathbf{q}(t)$ and the corresponding trajectory of the average opinion $M(t)$. The initial data employed in the top, middle, and bottom panels are $\mathbf{q}(0) = (0.31, 0.54, 0.05, 0.05, 0.05)$, $\mathbf{q}(0) = (0.1, 0.25, 0.3, 0.25, 0.1)$, and \\$\mathbf{q}(0) = (0.05, 0.05, 0.05, 0.549, 0.301)$, which yield initial average opinions $M(0) = 0.99 < 1$, $M(0) = 2$, and $M(0) = 3.001 > 3$, respectively. The numerical outcomes exhibit convergence consistent with the predictions of Proposition~\ref{prop:symmetric_IC} and Theorem~\ref{thm:basins_of_attraction}.

\def\figscale{0.5}
\def\figwidth{0.47\textwidth}
\begin{figure}
	\begin{subfigure}{\figwidth}
		\centering
		\includegraphics[scale=\figscale]{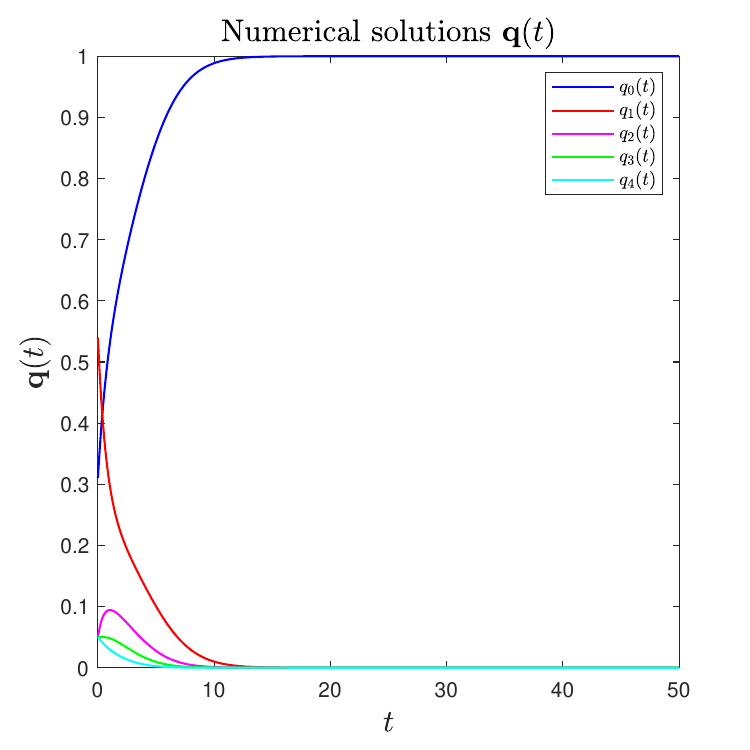}
	\end{subfigure}
	\hspace{0.1in}
	\begin{subfigure}{\figwidth}
		\centering
		\includegraphics[scale=\figscale]{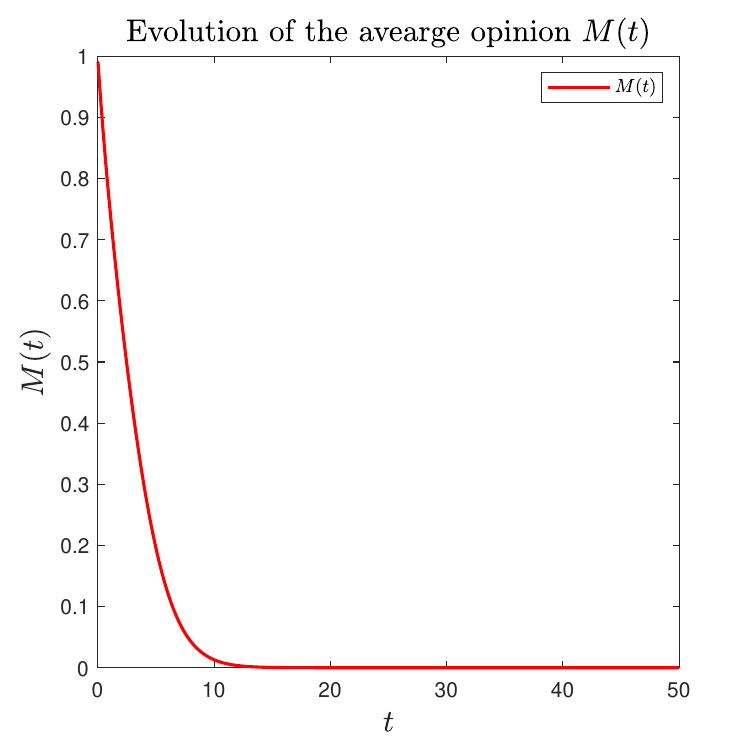}
	\end{subfigure}
	\\
	\begin{subfigure}{\figwidth}
		\centering
		\includegraphics[scale=\figscale]{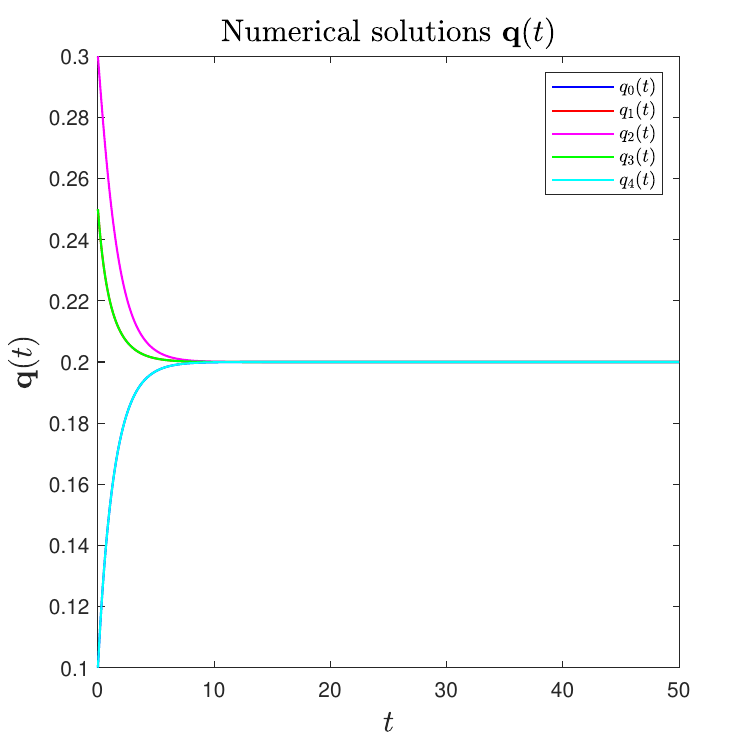}
	\end{subfigure}
	\hspace{0.1in}
	\begin{subfigure}{0.45\textwidth}
		\centering
		\includegraphics[scale=\figscale]{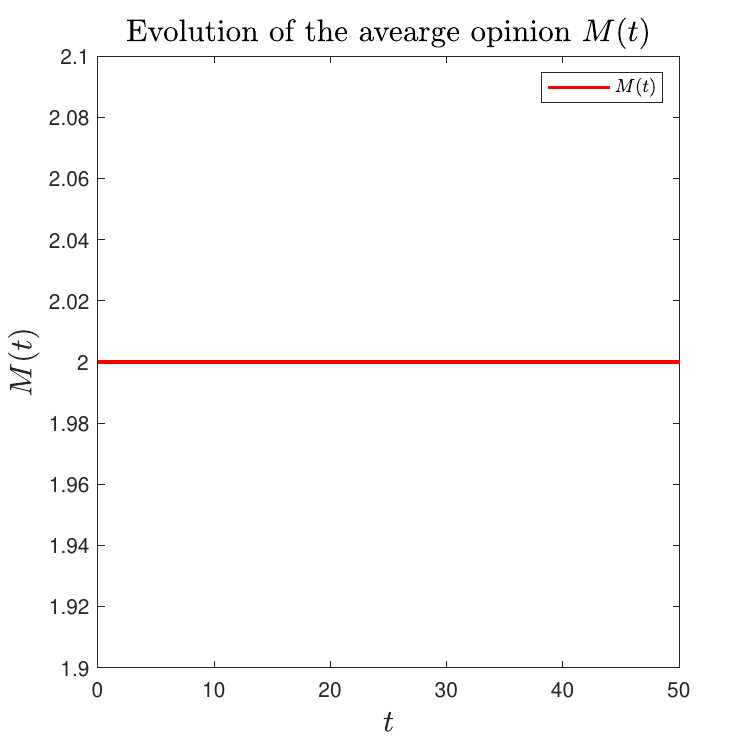}
	\end{subfigure}
    \\
    \begin{subfigure}{\figwidth}
    	\centering
    	\includegraphics[scale=\figscale]{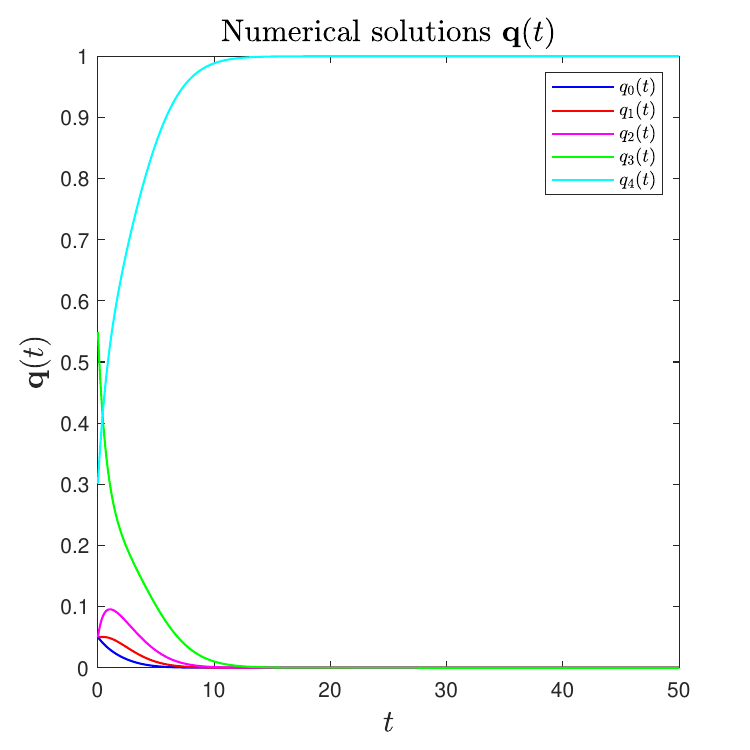}
    \end{subfigure}
    \hspace{0.1in}
    \begin{subfigure}{\figwidth}
    	\centering
    	\includegraphics[scale=\figscale]{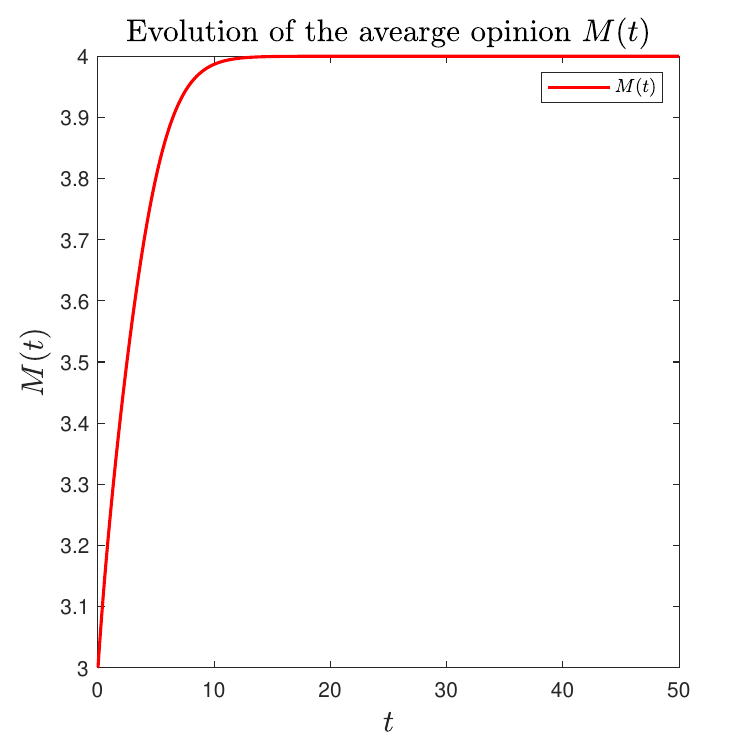}
    \end{subfigure}
	
\caption{{\bf Left}: Simulations of the mean-field dynamical system \eqref{eqn:ODE_main} with $k=2$, for three different initial datum with $M(0) < 1$ (top), symmetric initial opinion distribution (middle), and $M(0) > 3$ (bottom). {\bf Right}: Evolution of the corresponding average opinion with respect to time.}
\label{fig:numerics2}
\end{figure}

\begin{remark}
The proof of Theorem \ref{thm:basins_of_attraction} implies that the regions \[\Omega^l \coloneqq \left\{{\bf q} \in \cP(\Omega) \mid M < 1\right\} \subset \cA^l ~~\textrm{and}~~ \Omega^r \coloneqq \left\{{\bf q} \in \cP(\Omega) \mid M > 2k-1\right\} \subset \cA^r\] are actually invariant under the (forward) flow generated by the mean-field dynamical system \eqref{eqn:ODE_main}. As a concrete example, we have $(1/2,1/2,0,\cdots,0) \in \Omega^l$ and $(0,\cdots,0,1/2,1/2) \in \Omega^r$. On the other hand, determining which consensus profile (${\bf q}^l$ or ${\bf q}^r$) a generic initial datum ${\bf q}(0) \in \cP(\Omega) \setminus (\Omega^l \cup \Omega^r)$ is attracted to remains a challenging problem. Based on numerical experiments, we expect that $\lim_{t \to \infty} {\bf q}(t) = {\bf q}^l$ whenever $\sum_{n=0}^{k} q_n(0) = 1$ and $q_k(0) < 1$ (i.e., when initially no agents hold an opinion strictly larger than $k$, and the initial state is not already the consensus at opinion $k$). Similarly, we also expect that $\lim_{t \to \infty} {\bf q}(t) = {\bf q}^r$ whenever $\sum_{n=k}^{2k} q_n(0) = 1$ and $q_k(0) < 1$.
\end{remark}

\section{Conclusion}\label{sec:sec4}
\setcounter{equation}{0}

In this work, we propose and analyze a novel agent-based opinion dynamics in which each agent holds an opinion belonging to the discrete set of admissible states $\{-k, \ldots, 0, \ldots, k\}$, where $k \in \mathbb{N}_+$ is a prescribed model parameter. Our model adds to the growing body of opinion dynamics research in the sociophysics literature, featuring an interaction rule in which agents tend to adjust their opinions toward those of their interaction partners, with more extreme agents exerting stronger persuasive influence. We derive the corresponding mean-field limit, obtaining a coupled nonlinear system of ODEs that governs the large-population behavior of the model. This limiting system admits three equilibrium configurations/distributions, among which the two consensus states at the extreme opinions are stable. We also investigate the large-time behavior of the resulting mean-field opinion dynamics, focusing on the class of initial opinion profiles that satisfies suitable constraints on the initial average opinion. For such configurations, we establish precise convergence results and characterize the attracting consensus equilibria.

The present paper also leaves several important questions open for future investigation. First, can one obtain a complete and explicit description of the basins of attraction associated with the two consensus equilibria at the extreme opinions? In particular, given an arbitrary (non-symmetric) initial datum, is it possible to determine precisely which consensus state the solution of the mean-field opinion dynamics will converge to? Second, can we quantify the rate of convergence to equilibrium for a generic (non-symmetric) initial datum? While a quantitative convergence guarantee---such as Theorem~\ref{thm:conv_k=1}---can be obtained in the simplest case $k = 1$, deriving analogous quantitative estimates for general $k \geq 2$ appears to be considerably more challenging. A rigorous theoretical treatment of these problems would provide a deeper understanding of how the initial opinion profile, and potentially the model parameter $k$, influence the structure of the resulting equilibrium distribution of opinions. More broadly, such an analysis would contribute to the general theory of nonlinear dynamical systems with multiple stable equilibria, where the long-time behavior is highly sensitive to the geometry of the basins of attraction. Identifying these basins explicitly, or even partially, remains a central and often challenging task in the study of nonlinear evolution equations and multi-agent systems. In this sense, our model offers a concrete and tractable setting in which one can explore how microscopic interaction rules shape the global phase portrait, shedding light on how local biases or asymmetries in the initial data may steer the system toward qualitatively distinct macroscopic outcomes. Such insights could have broader implications not only for opinion dynamics, but also for other multi-agent systems in which competing stable states co-exist and the selection of an eventual equilibrium depends delicately on the initial configuration. \\

\noindent {\bf Acknowledgement~} Fei Cao gratefully acknowledges support from an AMS-Simons Travel Grant, administered by the American Mathematical Society with funding from the Simons Foundation. Xiaoqian Gong is supported by the National Science Foundation through Grant DMS-2418971.

\end{document}